\documentclass{acm_proc_article-sp}

\usepackage{amsmath,amssymb}
\usepackage{latexsym}
\usepackage{url}
\usepackage{xspace}
\usepackage{booktabs}
\usepackage{paralist}
\usepackage{tabulary}
\usepackage{subcaption}
\usepackage{tikz}
\usetikzlibrary{shapes,arrows,positioning,chains,decorations.pathreplacing,calc}

\tikzset{release/.style={regular polygon, regular polygon sides = 4,draw,inner sep=2pt}}
\tikzset{task/.style={circle,draw}}
\tikzset{gate/.style={regular polygon, regular polygon sides = 6,draw}}
\tikzset{exgate/.style={gate,fill=gray!10}}
\tikzset{ingate/.style={gate}}

\usepackage[ruled,lined,linesnumbered,boxed,vlined]{algorithm2e}
\SetKwInOut{Input}{input}
\SetKwInOut{Output}{output}

\usepackage{pgfplots}

\newcommand{\set}[1]{\left\{#1\right\}}

\newcommand{\etal}{{\em et al.}\xspace}
\newcommand{\vwsp}{\textsc{Valued WSP}\xspace}
\newcommand{\rrbac}{{\sf R$^2$BAC}}
\newcommand{\card}[1]{\left|#1\right|}
\newcommand{\ua}{\mathit{UA}}
\newcommand{\pa}{\mathit{PA}}

\newtheorem{thm}{Theorem}

\newtheorem{prop}{Proposition}

\newtheorem{df}{Definition}

\allowdisplaybreaks

\begin{document}

\title{Valued Workflow Satisfiability Problem}

\numberofauthors{3}

\author{
\alignauthor Jason Crampton\\
	   \affaddr{Information Security Group}
       \affaddr{Royal Holloway, University of London}\\
       \affaddr{Egham, Surrey}\\
       \affaddr{UK TW20\,0EX}\\
       \email{jason.crampton@rhul.ac.uk}
\alignauthor Gregory Z. Gutin\\
       \affaddr{Department of Computer Science}\\
       \affaddr{Royal Holloway, University of London}\\
       \affaddr{Egham, Surrey}\\
       \affaddr{UK TW20\,0EX}\\
       \email{g.gutin@rhul.ac.uk}
\alignauthor Daniel Karapetyan\\
       \affaddr{ASAP Research Group}\\
       \affaddr{School of Computer Science}\\
       \affaddr{University of Nottingham}\\
       \affaddr{Nottingham}\\
       \affaddr{UK NG8\,1BB}\\
       \email{g.gutin@rhul.ac.uk}
}

%\author{Jason Crampton \and Gregory Gutin \and Daniel Karapetyan}
 
\maketitle

\begin{abstract}
A workflow is a collection of steps that must be executed in some specific order to achieve an objective.
A computerised workflow management system may enforce authorisation policies and constraints, thereby restricting which users can perform particular steps in a workflow.
The existence of policies and constraints may mean that a workflow is unsatisfiable, in the sense that it is impossible to find an authorised user for each step in the workflow and satisfy all constraints.
In this paper, we consider the problem of finding the ``least bad'' assignment of users to workflow steps by assigning a weight to each policy and constraint violation.
To this end, we introduce a framework for associating costs with the violation of workflow policies and constraints and define the \emph{valued workflow satisfiability problem} (\vwsp), whose solution is an assignment of steps to users of minimum cost.
We establish the computational complexity of \vwsp with user-independent constraints and show that it is fixed-parameter tractable.
We then describe an algorithm for solving \vwsp with user-independent constraints and evaluate its performance, comparing it to that of an off-the-shelf mixed integer programming package.
%  Finally, we discuss the connections between \vwsp and related work, including risk-aware access control and constraint satisfaction problems in the presence of ``soft'' constraints.
\end{abstract}

% <concept_id>10002978.10002991.10010839</concept_id>
% <concept_desc>Security and privacy~Authorization</concept_desc>
% <concept_id>10002978.10002991.10002993</concept_id>
% <concept_desc>Security and privacy~Access control</concept_desc>
% <concept_significance>300</concept_significance>
% <concept_id>10002978.10003022.10003028</concept_id>
% <concept_desc>Security and privacy~Domain-specific security and privacy architectures</concept_desc>
% <concept_significance>100</concept_significance>
% <concept_id>10003752.10003809.10010052.10010053</concept_id>
% <concept_desc>Theory of computation~Fixed parameter tractability</concept_desc>
% <concept_significance>300</concept_significance>
% <concept_id>10002951.10003227.10003228.10003442</concept_id>
% <concept_desc>Information systems~Enterprise applications</concept_desc>
% <concept_significance>100</concept_significance>
% \ccsdesc[500]{Security and privacy~Authorization}
% \ccsdesc[300]{Security and privacy~Access control}
% \ccsdesc[100]{Security and privacy~Domain-specific security and privacy architectures}
% \ccsdesc[300]{Theory of computation~Fixed parameter tractability}
% \ccsdesc[100]{Information systems~Enterprise applications}

\category{D4.6}{Operating Systems}{Security and Protection}[Access controls]
\category{F2.2}{Analysis of Algorithms and Problem Complexity}{Nonnumerical Algorithms and Problems}
\category{H2.0}{Database Management}{General}[Security, integrity and protection]

\terms{Algorithms, Security, Theory}


\section{Introduction}\label{sec:intro}

It is increasingly common for organisations to computerise their business and management processes.
The co-ordination of the tasks or steps that comprise a computerised business process is managed by a workflow management system (or business process management system).
A workflow is defined by the steps in a business process and the order in which those steps should be performed.
A workflow is executed multiple times, each execution being called a \emph{workflow instance}.
Typically, the execution of each step in a workflow instance will be triggered by a human user, or a software agent acting under the control of a human user.
As in all multi-user systems, some form of access control, typically specified in the form of policies and constraints, should be enforced on the execution of workflow steps, thereby restricting the execution of each step to some authorised subset of the user population.

Policies typically specify the workflow steps for which users are authorised, what Basin \etal\ call \emph{history-independent} authorisations~\cite{BaBuKa12}.
Constraints restrict which groups of users can perform sets of steps.
It may be that a user, while authorised by the policy to perform a particular step $s$, is prevented (by one or more constraints) from executing $s$ in a specific workflow instance because particular users have performed other steps in the workflow (hence the alternative name of \emph{history-dependent} authorizations~\cite{BaBuKa12}).
The concept of a Chinese wall, for example, limits the set of steps that any one user can perform~\cite{BrNa89}, as does separation-of-duty, which is a central part of the role-based access control model~\cite{ansi-rbac04}.
We note that policies are, in some sense, discretionary, as they are defined by the workflow administrator in the context of a given set of users.
However, constraints may be mandatory (and independent of the user population), in that they may encode statutory requirements governing privacy or separation-of-concerns or high-level organisational requirements.

It is well known that a workflow specification may be ``unsatisfiable'' in the sense that the combination of policies and constraints means that there is no way of allocating authorised users to workflow steps without violating at least one constraint.
The workflow satisfiability problem is NP-hard~\cite{WaLi10} although relatively efficient algorithms have been developed on the assumption that the number of workflow steps is much smaller than the number of users that may perform steps in the workflow~\cite{CoCrGaGuJo14,CrGuYe13,KaGaGu,WaLi10}.
Of course, the objectives of the business process associated with a workflow specification can never be achieved if the specification is unsatisfiable.
Hence, it is interesting to consider an extended version of the workflow satisfiability problem that seeks the ``best'' allocation of users to steps in the event that the specification is unsatisfiable.

Accordingly, in this paper we study the \emph{valued} workflow satisfiability problem (\vwsp).
Informally, we associate constraint and authorisation violations with a cost, which may be regarded as an estimate of the risk associated with allowing those violations.
We then compute an assignment of users to steps having minimal cost, this cost being zero when the workflow is satisfiable.
In a sense, our work is related to recent work on risk-aware access control~\cite{ChCr11,ChRoKeKaWaRe07,DiBeEyBaMo04,NiBeLo10}, which seeks to compute the risk of allowing a user to perform an action, rather than simply computing an allow/deny decision, and ensure that cumulative risk remains within certain limits.
However, unlike related work, we focus on computing user-step assignments of minimal cost, rather than access control decisions.

Our main contributions are: 
  to define \vwsp and determine its complexity; 
  to prove that \vwsp is fixed-parameter tractable for weighted user-independent constraints; 
  %to provide an encoding of \vwsp as a mixed integer programming problem; 
  to develop an algorithm to solve \vwsp\ with user-independent constraints; 
  to provide a comprehensive experimental evaluation of our algorithm; 
  and to demonstrate that the performance of our algorithm compares very favourably with an approach that uses the mixed integer programming solver CPLEX.
Our experimental evaluation shows our algorithm enjoys a substantial advantage over CPLEX as the number of steps grows, with our algorithm being able to deal far better with problem instances containing more than 30 steps.
Moreover, our algorithm is far better at solving instances that are unsatisfiable -- precisely those instances for which \vwsp is relevant.

In the next section, we define \vwsp, having reviewed relevant concepts from the literature, including the workflow satisfiability problem and user-independent constraints.
In Section~\ref{sec:solving-vwsp}, we prove that \vwsp is fixed-parameter tractable for user-independent constraints and describe an algorithm based on the concept of a \emph{pattern}, which is, informally, a compact representation of a set of similar plans for a user-independent constraint.
In Section~\ref{sec:experiments}, we present our experimental results.
This section also includes a method of representing \vwsp as a mixed integer programming problem, which may be of use in subsequent research.
We conclude the paper with a discussion of related work, a summary of contributions and some suggestions for future work.
 
\section{Background and Problem\\ Statement}

We first briefly summarise relevant concepts from the literature, including \emph{workflow authorisation schema}, \emph{workflow constraints} and the \emph{workflow satisfiability problem} (WSP).
We then explain how a constrained workflow authorisation schema can be extended to assign costs to plans that do not satisfy the schema's policy and constraints.
We conclude this section with a formal definition of \vwsp.

\subsection{WSP}

A directed acyclic graph $G = (V,E)$ is defined by a set of nodes $V$ and a set of edges $E \subseteq V \times V$.
The reflexive, transitive closure of a directed acyclic graph (DAG) defines a partial order, where $v \leqslant w$ if and only if there is a path from $v$ to $w$ in $G$.
%If $(V,\leqslant)$ is a partially ordered set, then we write $v \parallel w$ if $v$ and $w$ are incomparable; that is, $v \not\leqslant w$ and $w \not\leqslant v$.
We may write $v \geqslant w$ whenever $w \leqslant v$.
We may also write $v < w$ whenever $v \leqslant w$ and $v \ne w$.
%Finally, we will write $[n]$ to denote $\set{1,\dots,n}$.

\begin{df}\label{def:workflow}
A \emph{workflow specification} is defined by a directed, acyclic graph $G = (S,E)$, where $S$ is a set of steps and $E \subseteq S \times S$.
Given a workflow specification $(S,E)$ and a set of users $U$, an \emph{authorisation policy} for a workflow specification is a relation $A \subseteq S \times U$.
A \emph{workflow authorisation schema} is a tuple $(G,U,A)$, where $G = (S,E)$ is a workflow specification and $A$ is an authorisation policy.
\end{df}

%We will use the representations of a workflow specification as a partial order and a DAG interchangeably.
A workflow specification describes a sequence of steps and the order in which they must be performed when the workflow is executed, each such execution being called a \emph{workflow instance}.%
%If $s < s'$ then $s$ must be performed before $s'$ in every instance of the workflow; if $s \parallel s'$ then $s$ and $s'$ may be performed in either order.%
\footnote{In this paper, the ordering on the steps is not considered.  Prior work has shown that the ordering is irrelevant to the question of satisfiability subject to certain assumptions about the constraints~\cite{CrGuYe13}, assumptions that are satisfied by the constraints considered in this paper.}
User $u$ is authorised to perform step $s$ only if $(s,u) \in A$.%
\footnote{In practice, the set of authorised step-user pairs, $A$, will not be defined explicitly.
          Instead, $A$ will be inferred from other access control data structures.
          In particular, \rrbac~--~the role-and-relation-based access control model of Wang and Li~\cite{WaLi10}~--~introduces a set of roles $R$, a user-role relation ${\it UR} \subseteq U \times R$ and a role-step relation ${\it SA} \subseteq R \times S$ from which it is possible to derive the steps for which users are authorised.
          For all common access control policies (including \rrbac), it is straightforward to derive $A$.
          We prefer to use $A$ in order to simplify the exposition.}
We assume that for every step $s \in S$ there exists some user $u \in U$ such that $(s,u) \in A$ (otherwise the workflow is trivially unsatisfiable).

\begin{df}
Let $((S,E),U,A)$ be a workflow authorisation schema.
A \emph{plan} is a function $\pi: T \rightarrow U$, where $T\subseteq S$. A plan $\pi$ is \emph{complete} if $T=S.$
%A plan $\pi$ is \emph{authorised} for $((S,E),U,A)$ if $(s,\pi(s)) \in A$ for all $s \in S$.
\end{df}

\begin{df}
A \emph{workflow constraint} has the form $(T,\Theta)$, where $T \subseteq S$ and $\Theta$ is a family of functions with domain $T$ and range $U$.
$T$ is the \emph{scope} of the constraint $(T,\Theta)$.
A \emph{constrained workflow authorization schema} is a pair $((S,E),U,A,C)$, where $((S,E),U,A)$ is a workflow authorization schema and $C$ is a set of workflow constraints.
\end{df}

Informally, a workflow constraint $(T,\Theta)$ limits the users that are allowed to perform a set of steps $T$ in any given instance of the workflow.
In particular, $\Theta$ identifies authorised (partial) assigments of users to workflow steps in $T$.
More formally, let $\pi :S' \rightarrow U$, where $S' \subseteq S$, be a plan.
Given $T \subseteq S'$, we write $\pi|_T$ to denote the function $\pi$ restricted to domain $T$; that is $\pi|_T(s) = \pi(s)$ for all $s \in T$ (and is undefined otherwise).
Then we say $\pi : S' \rightarrow U$ \emph{satisfies} a workflow constraint $(T,\Theta)$ if $T \not\subseteq S'$ or $\pi|_T \in \Theta$.
%In other words, if $S'$ includes all steps in $T$ then the allocation of users to steps in $T$ by $\pi$ must be a function in $\Theta$.
\begin{df}
Given a constrained workflow authorization schema $((S,E),U,A,C)$, we say a plan $\pi : S \rightarrow U$ is \emph{valid} if it satisfies every constraint in $C$ and, for all $t \in S$, $(t,\pi(t)) \in A$.
\end{df}

In practice, we do not define constraints by enumerating all possible elements of $\Theta$.
Instead, we define different families of constraints that have ``compact'' descriptions.
Thus, for example, we might define the family of simple separation-of-duty constraints, each of which is represented by a set $\set{t_1,t_2}$, the constraint being satisfied provided $t_1$ and $t_2$ are assigned to different users.

A constraint $(T,\Theta)$ is said to be \emph{user-independent} if, for every $\theta \in \Theta$ and every permutation $\phi : U \rightarrow U$, $\phi \circ \theta \in \Theta$, where 
\[
  \phi \circ \theta : T \rightarrow U\quad\text{and}\quad(\phi \circ \theta)(s) \stackrel{\rm def}{=} \phi(\theta(s)). 
\]
Simple separation-of-duty constraints are user-independent, and it appears most constraints that are useful in practice are user-independent~\cite{CoCrGaGuJo14}.
In particular, cardinality constraints and binding-of-duty constraints are user-independent.
In this paper, we restrict our experimental evaluation to two particular types of user-independent constraints (in addition to separation-of-duty constraints):
\begin{itemize}
 \item an \emph{at-least-$r$} \emph{counting} constraint has the form $(T,r)$, where $r \leqslant \card{T}$, and is satisfied provided at least $r$ users are assigned to the steps in $T$;
 \item an \emph{at-most-$r$} \emph{counting} constraint has the form $(T,r)$, where $r \leqslant \card{T}$, and is satisfied provided at most $r$ users are assigned to the steps in $T$.
\end{itemize}
It is important to stress that our approach works for any user-independent constraints.
We chose to use counting constraints because such constraints have been widely considered in the literature (often known as cardinality constraints).
Moreover, counting constraints can be encoded using mixed integer programming, so we can use off-the-shelf solvers to solve WSP and to compare with the performance of our bespoke algorithms.

We now introduce the workflow satisfiability problem, as defined by Wang and Li~\cite{WaLi10}.

\begin{center}
\fbox{%
      \begin{tabulary}{.95\columnwidth}{@{}r<{~}@{}L@{}}
        \multicolumn{2}{@{}l}{\sc Workflow Satisfiability Problem (WSP)}\\
        \emph{Input:} & A constrained workflow authorisation schema $((S,E),U,A,C)$\\
        \emph{Output:} & A valid 
      $\pi : S \rightarrow U$ or an answer that there exists no valid plan
       \end{tabulary}%
      }
\end{center}

\subsection{Fixed-Parameter Tractability of WSP}

A na\"ive approach to solving WSP would consider every possible assignment of users to steps in the workflow.
There are $n^k$ such assignments if there are $n$ users and $k$ steps, so an algorithm of this form would have complexity $O(mn^k)$, where $m$ is the number of constraints.
Moreover, Wang and Li showed that WSP is NP-hard, by reducing {\sc Graph $k$-Colorability} to WSP with separation-of-duty constraints~\cite[Lemma 3]{WaLi10}.

The importance of finding an efficient algorithm for solving WSP led Wang and Li to look at the problem from the perspective of \emph{parameterised complexity}.
Suppose we have an algorithm that solves an NP-hard problem in time $O(f(k)n^d)$, where $n$ denotes the size of the input to the problem, $k$ is some (small) parameter of the problem, $f$ is some function in $k$ only, and $d$ is some constant (independent of $k$ and $n$).
Then we say the algorithm is a \emph{fixed-parameter tractable} (FPT) algorithm.
If a problem can be solved using an FPT algorithm then we say that it is an \emph{FPT problem} and that it belongs to the class FPT~\cite{DoFe99,Ni06}.

Wang and Li observed that fixed-parameter algorithmics is an appropriate way to study the problem, because the number $k$ of steps is usually small and often much smaller than the number $n$ of users.\footnote{The SMV loan origination workflow studied by Schaad \etal, for example, has 13 steps and identifies five roles~\cite{ScLoSo06}. It is generally assumed that the number of users is significantly greater than the number of roles.} Wang and Li \cite{WaLi10}  proved that, in general, WSP is W[1]-hard and thus is highly unlikely to admit an FPT algorithm.
However, WSP is FPT if we consider only separation-of-duty and binding-of-duty constraints \cite{WaLi10}. Henceforth, we consider special families of constraints, but allow arbitrary authorisations.
Crampton \etal~\cite{CrGuYe13} obtained significantly faster FPT algorithms that were applicable to ``regular'' constraints, thereby including the cases shown to be FPT by Wang and Li.
Subsequent research has demonstrated the existence of FPT algorithms for {\sc WSP} in the presence of other constraint types~\cite{CrCrGuJoRa13,CrGu13}. Cohen \etal~\cite{CoCrGaGuJo14} 
introduced the class of user-independent constraints and showed that WSP remains FPT if only user-independent constraints are included. Note that every regular constraint is user-independent and all the constraints defined in the ANSI RBAC standard \cite{ansi-rbac04} are user-independent. Results of Cohen \etal~\cite{CoCrGaGuJo14}  have led to algorithms which significantly outperform the widely used SAT-solver SAT4J on difficult instances of WSP with user-independent constraints \cite{CoCrGaGuJo14c,KaGaGu}.

\subsection{Valued WSP}

There has been considerable interest in recent years in making the specification and enforcement of access control policies more flexible.
Naturally, it is essential to continue to enforce effective access control policies.
Equally, it is recognised that there may well be situations where a simple ``allow'' or ``deny'' decision for an access request may not be appropriate.
It may be, for example, that the risks of refusing an unauthorised request are less significant than the benefits of allowing it.
One obvious example occurs in healthcare systems, where the denial of an access request in an emergency situation could lead to loss of life.
Hence, there has been increasing interest in context-aware policies, such as ``break-the-glass'', which allow different responses to requests in different situations.
Risk-aware access control is another promising line of research that seeks to quantify the risk of allowing a request, where a decision of ``0'' might represent an unequivocal ``deny'' and ``1'' an unequivocal ``allow'', with decisions of intermediate values representing different levels of risk.

Similar considerations arise very naturally when we consider workflows.
In particular, we may specify authorization policies and constraints that mean a workflow specification is unsatisfiable.
Clearly, this is undesirable from a business perspective, since the business objective associated with the workflow can not be achieved.
There are two possible ways of dealing with an unsatisfiable workflow specification:%
\begin{inparaenum}[(i)]
 \item modify the authorization policy and/or constraints;
 \item find the ``least bad'' complete plan.
\end{inparaenum}
Prior work by Basin, Burri and Karjoth considered the former approach~\cite{BaBuKa12}.
They restricted their attention to modification of the authorisation policy, what they called \emph{administrable authorizations}.
They assigned costs to modifying different aspects of a policy and then computed a strategy to modify the policy of minimal cost.

We adopt a different approach and consider minimising the cost of ``breaking'' the policies and/or constraints.
(We will compare our approach to Basin \etal in the related work section.)
Informally, given a workflow specification, for each plan $\pi$, we define the total cost or weight associated with the plan $w(\pi)$.
The problem, then, is to find the complete plan with minimum total cost.

More formally, let $((S,E),U,A,C)$ be a constrained workflow authorization schema.
Let $\Pi$ denote the set of all possible plans from $S$ to $U$.
Then, for each $c \in C$, we define a weight function $w_c : \Pi \rightarrow \mathbb{Z}$, where
 \[ 
  w_c(\pi) 
   \begin{cases}
    = 0 & \text{if $\pi$ satisfies $c$}, \\
    > 0 & \text{otherwise}.
   \end{cases}
 \]
The pair $(c,w_c)$ is a \emph{weighted constraint}.

The intuition is that $w_c(\pi)$ represents the extent to which $\pi$ violates $c$.
Consider, for example, an at-most-$r$ counting constraint $(T,r)$.
Then $w_c(\pi)$ depends only on the number of users assigned to the steps in $T$ (and the penalty should increase as the number of users increases).
Let $\pi(T)$ denote the set of users assigned to steps in $T$.
Then $w_c(\pi) = 0$ if $\card{\pi(T)} \leqslant r$; for plans $\pi$ and $\pi'$, we have $w_c(\pi) = w_c(\pi')$ if $\card{\pi(T)} = \card{\pi(T')}$; and $0 < w_c(\pi) \le w_c(\pi')$ if $r < \card{\pi(T)} \le \card{\pi'(T)}$.
Similarly, for an at-least-$r$ constraint $c$ with scope $T$, we would have $w_{c}(\pi) = 0$ if $\card{\pi(T)} \geqslant r$; for plans $\pi$ and $\pi'$, we have $w_c(\pi) = w_c(\pi')$ if $\card{\pi(T)} = \card{\pi(T)}$; and $0 < w_{c}(\pi) \le w_{c}(\pi')$ if  $\card{\pi(T)} \ge \card{\pi'(T)} > 0$.

Then we define
\[
 w_C(\pi) = \sum_{c \in C} w_c(\pi),
\]
which we call the \emph{constraint weight} of $\pi$.
Note that $w_C(\pi) = 0$ if and only if $\pi$ satisfies all constraints in $C$.
Note also that $w_C(\pi)$ need not be defined to be the linear sum: $w_C(\pi)$ may be defined to be an arbitrary function of the tuple $(w_c(\pi):\ c\in C)$ and Theorem \ref{thm1} below would still hold. However, we will not use this generalisation in this paper, but simply remark that it is possible, if needed.

We next introduce a function $w_A : \Pi \rightarrow \mathbb{Z}$, which assigns a cost for each plan with respect to the authorisation policy.
The intuition is that a plan in which every user is authorised for the steps to which she is assigned has zero cost and the cost of a plan that violates the policy increases as the number of steps that are assigned to unauthorised users increases.
More formally, 
\[
 w_A(\pi) 
  \begin{cases}
   = 0 & \text{if $(\pi(t),t) \in A$ for all $t$}, \\
   > 0 & \text{otherwise}
  \end{cases}
\]
is the \emph{authorisation weight} of $\pi$.

The definition of $w_A$ can be arbitrarily fine-grained.
We could, for example, associate a weight $\omega(t,u)$ with every pair $(t,u)$, where a zero weight indicates that $u$ is authorised for $t$, and define
\[
 w_A(\pi) = \sum_{t \in T} \omega(t,\pi(t)).
\]
One particularly simple instantiation of this idea is to define a single weight $\omega > 0$ to be associated with every policy violation.
In this case, $w_A(\pi) = a \omega$, where $a$ is the number of steps that are assigned to unauthorised users.
Alternatively, we might distinguish between different types of users, so that, for example, assigning steps to external contractors is associated with a higher weight $\omega_e$ than the weight $\omega_i$ associated with assigning steps to (internal) unauthorised staff members.

We may now define the valued workflow satisfiability problem, which will be the subject of the remainder of the paper.

\begin{center}
\fbox{%
      \begin{tabulary}{.95\columnwidth}{@{}r<{~}@{}L@{}}
        \multicolumn{2}{@{}l}{\sc Valued WSP}\\
        \emph{Input:} & A constrained workflow authorisation schema $((S,E),U,A,C)$ with weights for constraints and authorisations, as above.\\
        \emph{Output:} & A plan $\pi : S \rightarrow U$ that minimises $w(\pi)=w_C(\pi) + w_A(\pi)$.
       \end{tabulary}%
      }
\end{center}
% 
% To extend authorisations, we define $w_A(T,u)$ for every set $T$ of steps and user $u\in U$ such that $w_A(T,u)=0$ if $t\in A(u)$ for each $t\in T$ and $w_A(T,u)>0$, otherwise. Thus, if $w_A(T,u)>0$ corresponds to the penalty of assigning some steps of $T$ to an unauthorised user $u.$ It follows from our definition that $w_A(\emptyset,u)=0$ for each $u$. Observe that for a user $u$, $\pi^{-1}(u)$ is the set of steps assigned to $u$ by $\pi$ (if no steps are assigned to $u$ then $\pi^{-1}(u)=\emptyset$). Let $w_A(\pi)=\sum_{u\in U} w_A(\pi^{-1}(u),u)$ be the {\em authorisation weight} of $\pi$. 

Before proceeding further, however, we introduce a weight function that is more fine-grained than those considered above, and the one that we shall use in the remainder of this paper.
Specifically, for each user $u$ and each subset $T$ of $S$, we define a weight $\omega(T,u) \in \mathbb{Z}$, where
\[
 \omega(T,u) 
  \begin{cases}
   = 0 & \text{if $(u,t) \in A$ for all $t \in T$}, \\
   > 0 & \text{otherwise}.
  \end{cases}
\]
We call $\omega: 2^S \times U$ the (\emph{weighted}) \emph{set-authorisation} function.
Vacuously, we have $\omega(\emptyset,u) = 0$ for all $u \in U$.
We write $\pi^{-1}(u)$ to denote the set of steps to which $u$ is assigned by plan $\pi$.
Then we define
\[
 w_A(\pi) = \sum_{u \in U}  \omega(\pi^{-1}(u),u).
\]
Clearly, this form of authorisation weight satisfies the required criteria.

We base $w_A$ on weights of the form $\omega(T,u)$ because, in addition to allowing us to specify weights for every pair $(t,u)$ if required, it allows us to express more complex (``non-linear'') costs on plans.
For example:
\begin{enumerate}
 \item We can introduce a large penalty $\omega(T,u)$, effectively saying we prefer not to involve $u$ in steps in $T$. (We use weights like this in our experimental work, described in Section~\ref{sec:experiments}.)
 \item We can define a limit $\ell$ on the number of steps that can be executed by $u$, by setting a large penalty $\omega(T,u)$ for all $T$ of cardinality greater than $\ell$.
 \item We can attempt to minimise the number of involved users by giving a small penalty for assigning a user to at least one step.  This is similar to item 1 above, albeit with a different goal.
 \item The weights associated with the same user executing different steps may not increase linearly.
 Once a user has performed one particular unauthorized step, the additional cost of executing a \emph{related} unauthorized step may be reduced, while the additional cost of executing an \emph{unrelated} unauthorized step may be the same as the original cost.
 Our formulation enables us to model this kind of situation.
%  \item The cost of licensing a user for two steps might be the same as licensing that user for either of those steps. 
 \item We can implement separation-of-duty on a per-user basis, which is not possible with user-independent constraints.  In particular, it may be acceptable for $u_1$ to perform steps $s_1$ and $s_2$, but not $u_2$, in which case $\omega(\{s_1,s_2\},u_1)$ would be small, while $\omega(\{s_1,s_2\},u_2)$ would be large.
\end{enumerate}

The next claim is an important observation following directly from the definitions above.

\begin{prop}\label{pro}
The optimal weight of an instance of {\sc Valued WSP}  equals zero if and only if the corresponding WSP instance is satisfiable.
\end{prop}

\section{Solving Valued WSP with User-Independent Constraints}\label{sec:solving-vwsp}

In Section \ref{sec:patterns} we introduce the notion of weighted user-independent constraints and prove that {\sc Valued WSP} with only user-independent constraints is fixed-parameter tractable (FPT). In Section \ref{sec:branch-and-bound}, we describe an FPT algorithm to solve  {\sc Valued WSP} with user-independent constraints.

\subsection{Weighted User-Independent Constraints and Patterns}\label{sec:patterns}

%Let us generalise user-independent constraints for {\sc Valued WSP}. 
A weighted constraint $c$ is called {\em user-independent}  if, for every permutation $\theta$ of $U$, $w_C(\pi)=w_C(\theta\circ\pi)$. Thus, a weighted user-independent constraint does not distinguish between users. 
Any (weighted) counting constraint for which the weight of plan $\pi$ is defined in terms of the cardinality of the image of $\pi$ is user-independent.

Given a plan $\pi:\ S'\rightarrow U$, where $S'\subseteq S$, the {\em pattern} $P(\pi)$ of $\pi$ is the partition $\{\pi^{-1}(u):\ u\in U,\ \pi^{-1}(u)\neq \emptyset\}$ of $S'$ into non-empty sets. We say that two plans $\pi$ and $\pi'$ are {\em equivalent} if they have the same pattern. If all constraints are user-independent and $\pi$ and $\pi'$ are equivalent, then $w_C(\pi)=w_C(\pi')$. A pattern is said to be {\em complete} if $S'=S$.

Generalising the corresponding result for WSP with user-independent constraints \cite[Theorem 2]{CoCrGaGuJo14}, we can prove the following theorem, which uses weighted set-authorisation. 
We assume%
\begin{inparaenum}[(i)]
 \item that the weight of each assignment can be computed in time polynomial in the number of steps, users and constraints (denoted $k$, $n$ and $m$, respectively);
 \item we can determine whether a plan satisfies a constraint in time polynomial in the number of steps and users.
\end{inparaenum}

\begin{thm}\label{thm1}
{\sc Valued WSP} in which all constraints are user-independent can be solved in time $2^{k\log k}(k+n+m)^{O(1)}$. Thus, {\sc Valued WSP} with user-independent constraints is FPT. 
\end{thm}
\begin{proof}
For a positive integer $x$, let $[x]=\{1,\dots ,x\}.$ 
%Let $S=\{s_i:\ i\in [k]\}.$ 
Recall that for equivalent complete plans $\pi$ and $\pi'$, we have $w_C(\pi)=w_C(\pi').$ 
However, $w_A(\pi) = \sum_{u\in U} w_A(\pi^{-1}(u),u)$ is, in general, different from $w_A(\pi')$ and so we must compute the minimum value of $w_A(\pi)$ among all equivalent complete plans $\pi$. 
To do so efficiently, for a complete plan $\pi$, we first construct a weighted complete bipartite graph $G_{\pi}$ with partite sets $[p]$ and $U$, where $p=|P(\pi)|$ as follows. Let $P(\pi)=\{T_1,\ldots, T_p\}.$ The weight of an edge $\{q,u\}$ is $\sum_{s\in T_q}\omega(s,u).$

Now observe that $G_{\pi'}=G_{\pi}$ for every pair $\pi,\pi'$ of equivalent complete plans and that $w_A(\pi)$ equals the weight of the corresponding matching of $G_{\pi}$ covering all vertices of $[p]$. Hence, it suffices to find such a matching of $G_{\pi}$ of minimum weight, which can be done by the Hungarian method~\cite{Ku55} in time $O(n^3)$.

Observe that the number of partitions of the set $[k]$ into non-empty subsets, called the Bell number $B_k$, is smaller than $k!$ and there are algorithms of running time $O(B_k)=O(2^{k\log k})$  to generate all partitions of $[k]$ \cite{Er88}. Thus, we can generate all patterns in time $O(2^{k\log k})$. For each of them we compute the corresponding complete plan of minimum weight, and, among all such plans,  we choose the one of smallest weight. The total running time is  $O(2^{k\log k}(k+n+m)^{O(1)})$.\end{proof}

Cohen \etal proved~\cite{CoCrGaGuJo14} that WSP with user-independent constraints cannot be solved in time $2^{o(k\log k)}(k+n+m)^{O(1)}$, unless the widely believed Exponential Time Hypothesis\footnote{The Exponential Time Hypothesis claims there is no algorithm of running time $O^*(2^{o(n)})$ for 3SAT on $n$ variables~\cite{ImPaZa01}.} fails. This and Proposition~\ref{pro} imply that the FPT algorithm of Theorem \ref{thm1} is optimal, in a sense.

\subsection{Pattern Branch and Bound}\label{sec:branch-and-bound}

We present a branch-and-bound algorithm (Algorithm~\ref{alg:pb}) for the \textsc{Valued WSP}\@.
This algorithm is inspired by the Pattern Backtracking algorithm for the WSP~\cite{KaGaGu}. However, the original algorithm solves a decision problem, whereas the \textsc{Valued WSP} is an optimisation problem.
Thus our algorithm for \vwsp requires a completely different algorithmic framework.
We call our algorithm for the \textsc{Valued WSP} Pattern Branch and Bound (PBB).
Given a pattern $P$, we will write $T(P) = \bigcup_{p \in P} p$ to denote the set $T \subseteq S$ on which $P$ is the partition.

%node extends its parent with one step such that  the algorithm selected a step $s \in S$ such that $s$ is not in $p$, and produces the child nodes $p_1$, $p_2$, \ldots, $p_{|p| + 1}$ by extending 

\begin{algorithm}[htb]
\caption{Entry point of the PBB algorithm.}
\label{alg:pb}

\Input{\textsc{Valued WSP} instance $W$}
\Output{The optimal complete plan for $W$}

$\lambda \gets \lambda(W)$, where $\lambda(W)$ is the global lower bound\; \label{line:global-lb}
$\pi \gets H(W)$, where $H(W)$ is a \textsc{Valued WSP} heuristic\; \label{line:heuristic}
\Return {$\mathit{explore}(\emptyset, \pi, \lambda)$}\;
\end{algorithm}

\begin{algorithm}[htb]
\caption{Recursive procedure $\mathit{explore}(P, \pi^*, \lambda)$ of the PBB algorithm.}
\label{alg:explore}

\Input{Node $P$ of the search tree; best plan $\pi^*$ found so far; the global lower bound $\lambda$}
\Output{The best plan found after exploring this branch of search}

\eIf {$P$ is a complete pattern}
{
    Let $\pi'$ be an optimal plan with pattern $P$\;
    $\pi'' \gets \arg\min_{\pi \in \{ \pi', \pi^* \}} w(\pi)$\;
}
{
    $\pi'' \gets \pi^*$\;
    %Let $P = \{ p_1, p_2, \ldots, p_{|P|} \}$\;
    Select $s \in S \setminus T(P)$ maximising $I(P, s)$\; \label{line:select-s}
    \ForAll {extensions $P'$ of $P$ with step $s$}
    {
        \If {$L(P') < w(\pi'')$ and $w(\pi'') > \lambda$}
        {
            Let $\pi' \gets \mathit{explore}(P', \pi'', \lambda)$\;
            $\pi'' \gets \arg\min_{\pi \in \{ \pi', \pi'' \}} w(\pi)$\;
        }
    } \label{line:extend-pattern}
    
%    \For {$i = 1, 2, \ldots, |P| + 1$}
%    {
%        Produce $P'$ from $P$ by extending $p_i$ with $s$ (if $i = |P| + 1$, add a new set $\{ s \}$ to $P'$)\; \label{line:extend-pattern}
%        \If {$L(P') < w(\pi'')$ and $w(\pi'') > \lambda$}
%        {
%            Let $\pi' \gets \mathit{explore}(P', \pi'', \lambda)$\;
%            $\pi'' \gets \arg\min_{\pi \in \{ \pi', \pi'' \}} w(\pi)$\;
%        }
%    }
}

\Return {$\pi''$}\;
\end{algorithm}

The general idea of the algorithm follows the proof of Theorem~\ref{thm1}.
The algorithm explores the space of patterns with a depth-first search and for each complete pattern $P$ seeks an optimal plan $\pi$ such that $P(\pi) = P$ (recall that such a plan can be found efficiently).
Each node of the search tree is a pattern, with the root being an empty pattern and leaves being complete patterns.
In each non-leaf node $P$, the algorithm selects a step $s \in S$ such that $s \not\in T(P)$ (line~\ref{line:select-s} of Algorithm~\ref{alg:explore}), and generates one child node for each possibility to extend $P$ with $s$ (line~\ref{line:extend-pattern}).
By extensions of $P$ with step $s$ we mean patterns $P'$ obtained from $P$ by adding $s$ to one of the subsets $p \in P$ or adding a new subset $\{ s \}$ to the partition; hence, there are $|P| + 1$ extensions of $P$.

Like any branch-and-bound algorithm, PBB utilises a lower bound $L(P)$ for pruning branches.
The lower bound $L(P)$ in a node $P$ is computed as follows:
$$
L(P) = \sum_{c \in C} L_c(P) + \sum_{p \in P} \min_{u \in U} \omega(p, u) \,,
$$
where $L_c(P)$ is the lower bound of $w_c(\pi)$, where $\pi$ is an extension of a plan with pattern $P$.
The implementation of $L_c(P)$ depends on the constraint type.
For example, for a counting constraint $c$ with the scope $T$ and weight function $w_c(\pi) = \omega_c(|\pi(T)|)$, the lower bound can be computed as $L_c(P) = l(q, a)$, where $q = |\{ p \in P :\ p \cap T \neq \emptyset \}|$, \mbox{$a = |T \cap T(P)|$} and $l(q, a)$ is the following recursive function:
\[
    l(q, a) = \begin{cases}
        \omega_c(q) & \text{if } a = |T|, \\
        \min \{ l(q, a + 1), l(q + 1, a + 1) \}
            & \text{otherwise}.
    \end{cases} %\\
  %  &= \min \{ \omega_c(i) : t \le i \le  \min \{ |T_c|, q + |T_c| - a \} \} \,.
\]

Other speed-ups implemented in the PBB algorithm include a global lower bound $\lambda(W)$ (line~\ref{line:global-lb} of Algorithm~\ref{alg:pb}) and the heuristic $H(W)$ to obtain a good upper bound from the very beginning of the search (line~\ref{line:heuristic} of Algorithm~\ref{alg:pb}).
In a simple implementation, the global lower bound $\lambda(W)$ could be a constant function $\lambda(W) = 0$; that would terminate the algorithm as soon as a complete plan satisfying all the constraints and authorisations is found.
The heuristic algorithm can be as simple as a trivial plan assigning some user to all the steps $S$.

In our implementation, however, we translate the \textsc{Valued WSP} instance $W$ into WSP instances and solve them to obtain better global lower bound and upper bound.
Let $\text{WSP}(W, x)$ be a WSP instance obtained from $W$ by eliminating all the constraints and authorisations with penalties below $x$ and converting the rest of the constraints and authorisations into hard constraints.
By solving $\text{WSP}(W, x)$ we establish either the global lower bound or the upper bound.
If $\text{WSP}(W, x)$ is unsatisfiable, we conclude that there exists no complete plan $\pi$ such that $w(\pi) < x$ and, hence, $\lambda(W) = x$.
Otherwise, the plan valid in $\text{WSP}(W, x)$ can be used for an upper bound in $W$.
We start from solving the $\text{WSP}(W, 1)$ and, if it turns out to be unsatisfiable, we solve the $\text{WSP}(W, M)$, where $M$ is a large enough number, which usually gives us a good upper bound as it rules out plans that break highly-penalised constraints.

The order in which patterns are extended (a step at a time) makes no difference to the worst-case time complexity of the algorithm but is crucial to its  performance in practice~\cite{KaGaGu}.
It is defined by the `importance' function $I(P, s)$, the intention being to focus on the most important steps as early as possible to quickly prune fruitless branches of the search.
The importance of a step mostly depends on the constraints that include the step in their respective scopes.
For example, if a step is involved in several separation-of-duty constraints, adding it to the pattern may significantly reduce the search space and possibly result in increased penalties for some at-most constraints.
Another example is if most of the steps of some constraint's scope are assigned, and adding the remaining steps to the pattern may have a severe effect on the penalty for that constraint.
The `importance' metric is context-dependent, i.e.\ the order of steps needs to be determined dynamically in each branch of the search tree.

The `importance' function $I(P,s)$ is a heuristic expression which we parametrised and optimised by an automated parameter tuning method.
Our function $I(P, s)$ takes into account the number and types of the constraints in which the step is involved.
In addition, it accounts for the constraints with incomplete scopes.
%For each constraint type, we use a matrix of `importance' weights with coordinates corresponding to the number of already assigned users and the number of assigned steps from the scope of the constraint.
Finally, we check intersections of `conflicting' constraints such as at-most and not-equals or at-most and at-least.

%, and then every child of a node $p$ is an extension of $p$ obtained by adding exactly one step to $p$.
%The child nodes of a partial pattern $p$ are obtained by extending $p$ with one more step $s \in S$.

%Each child node $p'$ of a parent $p$ is obtained by adding one step to $p$
%The root of the search tree corresponds to an empty partial pattern.
%Each of the other nodes of the search tree extends the parent partial pattern by either estending one of .
%The order in which the steps are assigned is determined dynamically and will be discussed below.

As shown in the proof of Theorem~\ref{thm1}, finding the optimal assignment of users given a fixed pattern can be done in $O(n^3)$ time (if computing $w_C(\pi)$ takes $O(n^3)$ time and computing $\omega(T,u)$ takes $O(n^3 / k)$ time).
Each non-leaf node of the search tree has at least two child nodes and, hence, the size of the search tree is within $O(B_k)$.
Then the worst case time complexity of the PBB algorithm is $O(B_k n^3)$.
%To improve the average running times of the algorithm, we use the branch and bound technique: if in some node of the search tree the lower bound reaches the upper bound, the whole branch of the tree can be pruned immediately.

%In our implementation, we compute the `importance' of each step in every node and then pick the most important step to branch on it.  The `importance' function is a heuristic expression which we parametrised and optimised by an automated parameter tuning method.  Our `importance' function takes into account the number and types of the constraints in which the step is involved.  In addition, it accounts for the constraints with incomplete scopes.  For each constraint type, we use a matrix of `importance' weights with coordinates corresponding to the number of already assigned users and steps to the scope of the constraint.  Finally, we check intersections of `conflicting' constraint such as at-most and not-equals or at-most and at-least.

With the exception of the `step importance' function $I(P, s)$, which is easy to adjust for any type of instances, and the lower bound $L(P)$, our algorithm is a generic solver for the user-independent \textsc{Valued WSP}\@.
For example, it does not exploit the specifics of the counting constraints which could be used to preprocess problem instances~\cite{CoCrGaGuJo14}.
This shows that our approach is generic, easy to implement and its performance can be further improved by implementing instance-specific speed-ups.

\section{Experimental Results}\label{sec:experiments}

The pseudo-Boolean SAT solver SAT4J has been used to solve the \textsc{WSP}~\cite{WaLi10}.
Recent work has demonstrated that a bespoke pattern-based algorithm can outperform SAT4J in solving WSP~\cite{CoCrGaGuJo14c,KaGaGu}.
Integer linear programming has been used by Basin \etal to solve the \emph{allocation existence problem}~\cite{BaBuKa12}, which is related to \vwsp.  
In this section, we describe the experimental work on \vwsp that we have undertaken.
In particular, we will compare the performance of our PBB algorithm to that of the state-of-art commercial MIP solver in our computational experiments on \vwsp.  
We first describe the problem instances we used and how we represented \vwsp as a mixed integer programming (MIP) problem.
We then present the results of our experiments.

\subsection{Benchmark Instances}
\label{sec:benchmark}

We use a pseudo-random instance generator to produce benchmark instances.
Our generator is an extension of an existing instance generator for WSP~\cite{KaGaGu}.
The parameters of the generator are the number $k$ of steps in the instance, the not-equals constraints density $d$ in the range 0--100\%, the multiplier $\alpha$ for the number of constraints and the seed value for initialisation of the pseudo-random number generator.
The generator produces an instance with $k$ steps and $10k + 10$ users: $10k$ employees $U'$ and 10 external consultants $U''$\@.
The penalty for assigning steps $T \subseteq S$ to an employee $u \in U'$ is given by
$$
\omega(T,u) = |T \cap B| \cdot 10 + |T \setminus (A \cup B)| \cdot 10^6 \,,
$$
where $A \subset S$ and $B \subset S$ are selected uniformly at random from $S$, with $A\cap B=\emptyset$ and $|A|$  being selected uniformly at random from $[1, \lceil (k-4)/2 \rceil ]$, and $|B| = 2$.
The penalty for assigning steps $T \subseteq S$ to an external consultant $u \in U''$ is given by
$$
\omega(T,u) = \begin{cases}
0 & \text{if } T = \emptyset, \\
20 & \text{if } T \neq \emptyset \text{ and } T \subseteq A, \\
10^6 \cdot |T \setminus A| & \text{if } T \subseteq S \setminus A, \\
10^6 \cdot |T \setminus A| + 20 & \text{otherwise},
\end{cases}
$$
where $A \subset S$ is selected uniformly at random having selected $|A|$ uniformly at random from $[1, \lceil k/4 \rceil]$.

Further, $\left\lfloor (d k(k - 1) + 1)/2 \right\rfloor$ distinct not-equals constraints are produced uniformly at random, each with penalty $10^6$ for assigning one user to both steps.
Finally, $\alpha k$ at-most-3 constraints and $\alpha k$ at-least-3 constraints are generated uniformly at random.
The scopes of all the at-most-3 and at-least-3 constraints are set to 5 steps.
The at-least-3 penalties are defined as $\omega_c(1) = 10^6$, $\omega_c(2) = 1$, $\omega_c(3) = \omega_c(4) = \omega_c(5) = 0$.
The at-most-3 penalties are defined as $\omega_c(1) = \omega_c(2) = \omega_c(3) = 0$, $\omega_c(4) = 5$ and $\omega_c(5) = 10$.

\subsection{Mixed Integer Programming Formulation}
\label{sec:mip}

In order to use an MIP solver, we propose an efficient MIP formulation of the \vwsp.
Note that the MIP formulation is specific to the particular constraints present in the instances, unlike the PBB algorithm.  
In this section we describe an MIP formulation for the instances described in Section~\ref{sec:benchmark}.
%In this section, we consider {\sc Valued WSP} which has only three types of constraints: not-equals constraints, at-most-$r$ constraints, and at-least-$r$ constraints  (the value of $r$ may be different for different constraints).
%We will show that such {\sc Valued WSP} can be formulated as a mixed integer linear program.

%We will need the following notations.
Let $C = C_{\le} \cup C_{\ge}$, where $C_{\le}$ is the set of at-most-$r$ constraints and $C_{\ge}$ is the set of at-least-$r$ constraints.
(Note that not-equals constraints can be modelled as at-least-$2$ constraints with the scope of two steps.)
For each constraint $c \in C$ we are given its scope $T_c \subseteq S$, %the penalty $w_c(q) \ge 0$ for assigning $q$ distinct users to its scope $T_c$.
%We require that $w_c(q) = 0$ if the constraint $c$ is satisfied and $w_c(q) > 0$ otherwise.
the minimum (maximum, respectively) number $r_c$ of users that can be assigned to $c \in C$ such that the at-most (at-least, respectively) constraint $c$ is satisfied and the penalty $\omega_c(q)$ for assigning $q$ distinct users to $T_c$ (note that $\omega_c(r_c) = 0$).
%For each at-most-r constraint $c \in C_{\le}$, we are given the maximum number of users $r_c$ that can be assigned to it without causing any penalty and the weight $w_{cq}$ of assigning $q > r_c$ users to it.
%For each at-least constraint $c \in C_{\ge}$, we are given the minimum number of users $r_c$ that can be assigned to it without causing any penalty and the weight $w_c(q)$ of assigning $q < r_c$ users to it.
%We assume that $w_{c,r_c} = 0$ for $c \in C_{\le} \cup C_{\ge}$.

%Let $U = U_\text{emp} \cup U_\text{cons}$, where $U_\text{emp}$ is the set of employees of the organisation, and $U_\text{cons}$ be the set of external consultants.
For each employee $u \in U'$ and each step $s \in S \setminus A(u)$ we are given an additive weight $\omega_{su} > 0$ of assigning $u$ to $s$, which models the penalties for steps in both $B(u)$ and $S \setminus (A(u) \cup B(u))$.
For each consultant $u \in U''$ we are given a set of steps $A(u) \subseteq S$, any of which can be assigned to $u$ for a penalty $\omega_u > 0$, and a weight $\Omega_u > 0$ for assigning a step $s \in S \setminus A(u)$ to $u$.

The complete plan in our formulation is defined by binary decision variables $x_{su}$, $s \in S$, $u \in U$.
Variable $x_{su}$ takes value 1 if step $s$ is assigned to user $u$ and 0 otherwise.
The \vwsp is then encoded in~(\ref{eq:mip-objective})--(\ref{eq:mip-z}):
\begin{flalign}
\label{eq:mip-objective}
\text{minimise}~&
 \sum_{c \in C_{\le}} \sum_{q = r_c + 1}^{|T_c|} (\omega_c(q) - \omega_c(q - 1)) p_{cq}
\nonumber &&\\
&	+ \sum_{c \in C_{\ge}} \sum_{q = 1}^{r_c - 1} (\omega_c(q) - \omega_c(q + 1)) p_{cq} \nonumber &&\\
&	+ \sum_{u \in U'} \sum_{s \in S \setminus A(u)} \omega_{su} x_{su} \nonumber &&\\
&	+ \sum_{u \in U''} \omega_u z_u
	+ \sum_{u \in U''} \Omega_u \sum_{s \in S \setminus A(u)} x_{su} &&
\end{flalign}
subject to
\begin{flalign}
\label{eq:mip-one-user-per-step}
& \sum_{u \in U} x_{su} = 1 \text{ for } s \in S, &&\\
\label{eq:mip-at-most}
& \sum_{u \in U} y_{cu} - \sum_{q = r_c + 1}^{|T_c|} p_{cq} \le r_c \text{ for each } c \in C_{\le}, &&\\
\label{eq:mip-at-most-order}
& p_{cq} - p_{c,q+1} \ge 0 \text{ for } c \in C_{\le} \text{ and } q = r_c + 1, \ldots, |T_c| - 1, &&\\
\label{eq:mip-at-least}
& \sum_{u \in U} y_{cu} + \sum_{q = 1}^{r_c - 1} p_{cq} \ge r_c \text{ for each } c \in C_{\ge}, &&\\
\label{eq:mip-at-least-order}
& p_{cq} - p_{c,q+1} \le 0 \text{ for } c \in C_{\ge} \text{ and } q = 1, 2, \ldots, r_c - 2, &&\\
\label{eq:mip-at-most-y}
& y_{cu} - x_{su} \ge 0 \text{ for each } c \in C_{\le},\ u \in U \text{ and } s \in T_c, &&\\
\label{eq:mip-at-least-y}
& y_{cu} - \sum_{s \in T_c} x_{su} \le 0 \text{ for each } c \in C_{\ge} \text{ and } u \in U, &&\\
\label{eq:mip-z-constraint}
& z_u - x_{su} \ge 0 \text{ for each } u \in U'' \text{ and } s \in A(u), &&\\
\label{eq:mip-x}
& x_{su} \in \{ 0, 1 \} \text{ for } s \in S \text{ and } u \in U, &&\\
\label{eq:mip-y}
& 0 \le y_{cu} \le 1 \text{ for } c \in C \text{ and } u \in U, &&\\
\label{eq:mip-p-at-most}
& p_{cq} \in \{ 0, 1 \} \text{ for } c \in C_{\le} \text{ and } q = r_c + 1, \ldots, |T_c|, &&\\
\label{eq:mip-p-at-least}
& p_{cq} \in \{ 0, 1 \} \text{ for } c \in C_{\ge} \text{ and } q = 1, 2, \ldots, r_c - 1, &&\\
\label{eq:mip-z}
& 0 \le z_u \le 1 \text{ for } u \in U''.&&
\end{flalign}

In addition to binary variables $x_{su}$, we introduce some other variables.
Binary variables $y_{cu}$, $c \in C$, $u \in U$ determine if user $u$ is assigned to some steps in the scope $T_c$ of constraint $c$.
Since $y_{cu}$ for $c \in C_{\le}$ is minimised and it is limited from below by binary expressions (\ref{eq:mip-at-most-y}), its integrality constraint can be waived.
Since $y_{cu}$ for $c \in C_{\ge}$ is maximised and it is limited from above by binary expressions (\ref{eq:mip-at-least-y}), its integrality constraint can also be waived.
Similar logic applies to $z_u$, which indicates if the consultant $u \in U''$ is assigned any steps in $B(u)$.
Finally, we introduce the binary variables $p_{cq}$ for each $c \in C$ and $q \in \mathbb{N}$ such that $w_c(q) > 0$.
These variables are responsible for the constraint penalties and (with appropriate limitations imposed on the instances, as our instance generator does) the integrality of
$p_{cq}$ and constraints~\eqref{eq:mip-at-most-order} and~\eqref{eq:mip-at-least-order} can be waived.

The objective function~(\ref{eq:mip-objective}) is the weight of the plan defined by $x_{su}$, and thus our aim is to minimise it.

\subsection{Experimental Results}
\label{sec:experimental-results}

We conducted a series of computational experiments to test the performance of the \textsc{Valued WSP} solution methods.
Our test machine is powered by two Intel Xeon CPU's E5-2630~v2 (2.6~GHz) and has 32 GB RAM installed.
The PBB algorithm is implemented in C\#, and the MIP formulation is solved with CPLEX~12.6.
In all our experiments, each solver run is allocated exactly one physical CPU core.
Each result is reported as an average over 100 runs for 100 instances obtained by changing the random generator seed value.

\begin{table*}[htb] \centering
\begin{tabular}{@{} r r r @{} c @{} r @{} c @{} r @{} c @{} r r @{} c @{} r r @{} c @{} r r @{} c @{} r r @{} c @{} r r @{}}
\toprule
&&&&&&&&\multicolumn{2}{@{}c@{}}{Solved}&&\multicolumn{2}{@{}c@{}}{Time, sec}&&\multicolumn{2}{@{}c@{}}{$w_C(\pi)$}&&\multicolumn{2}{@{}c@{}}{$w_A(\pi)$}&&\multicolumn{2}{@{}c@{}}{Best $w(\pi)$}\\
\cmidrule(){9-10}
\cmidrule(){12-13}
\cmidrule(){15-16}
\cmidrule(){18-19}
\cmidrule(){21-22}
$k$&$d$&$\alpha$&\hspace*{2em}&Sat.&\hspace*{2em}&$w(\pi)$&\hspace*{2em}&PBB&MIP&\hspace*{2em}&PBB&MIP&\hspace*{2em}&PBB&MIP&\hspace*{2em}&PBB&MIP&\hspace*{2em}&PBB&MIP\\
\midrule
20&10\%&0.50&&100\%&&0.0&&100\%&100\%&&\underline{0.0}&5.9&&0.0&0.0&&0.0&0.0&&---&---\\
20&20\%&0.50&&90\%&&0.4&&100\%&100\%&&\underline{0.0}&19.8&&0.4&0.4&&0.0&0.0&&---&---\\
20&30\%&0.50&&37\%&&4.0&&100\%&100\%&&\underline{0.0}&65.0&&3.8&4.0&&0.2&0.0&&---&---\\
20&10\%&1.00&&18\%&&4.4&&100\%&100\%&&\underline{0.1}&556.0&&4.3&4.2&&0.1&0.2&&---&---\\
20&20\%&1.00&&0\%&&14.2&&100\%&100\%&&\underline{0.1}&532.9&&13.5&13.8&&0.7&0.4&&---&---\\
20&30\%&1.00&&0\%&&24.3&&100\%&100\%&&\underline{0.1}&469.9&&23.4&23.5&&0.9&0.8&&---&---\\
[1ex]
25&10\%&0.50&&100\%&&0.0&&100\%&100\%&&\underline{0.0}&32.0&&0.0&0.0&&0.0&0.0&&---&---\\
25&20\%&0.50&&93\%&&0.4&&100\%&100\%&&\underline{0.0}&102.2&&0.4&0.4&&0.0&0.0&&---&---\\
25&30\%&0.50&&27\%&&5.0&&100\%&100\%&&\underline{0.0}&319.3&&5.0&5.0&&0.0&0.0&&---&---\\
25&10\%&1.00&&40\%&&2.3&&100\%&39\%&&0.3&?&&2.3&?&&0.0&?&&---&4.1\\
25&20\%&1.00&&0\%&&14.3&&100\%&66\%&&0.8&?&&13.7&?&&0.6&?&&---&14.9\\
25&30\%&1.00&&0\%&&29.9&&100\%&95\%&&2.0&?&&28.9&?&&1.0&?&&---&29.9\\
[1ex]
30&10\%&0.50&&100\%&&0.0&&100\%&100\%&&\underline{0.0}&72.2&&0.0&0.0&&0.0&0.0&&---&---\\
30&20\%&0.50&&88\%&&0.6&&100\%&99\%&&0.0&?&&0.6&?&&0.0&?&&---&0.6\\
30&30\%&0.50&&24\%&&5.7&&100\%&99\%&&0.1&?&&5.7&?&&0.0&?&&---&5.7\\
30&10\%&1.00&&6\%&&5.8&&100\%&3\%&&5.1&?&&5.7&?&&0.1&?&&---&14.9\\
30&20\%&1.00&&0\%&&22.8&&100\%&7\%&&36.3&?&&22.4&?&&0.4&?&&---&31.5\\
30&30\%&1.00&&0\%&&43.5&&100\%&31\%&&173.7&?&&43.2&?&&0.3&?&&---&52.8\\
[1ex]
35&10\%&0.50&&100\%&&0.0&&100\%&100\%&&\underline{0.0}&195.9&&0.0&0.0&&0.0&0.0&&---&---\\
35&20\%&0.50&&91\%&&0.4&&100\%&89\%&&0.2&?&&0.4&?&&0.0&?&&---&0.6\\
35&30\%&0.50&&13\%&&7.7&&100\%&68\%&&18.5&?&&7.7&?&&0.0&?&&---&8.0\\
35&10\%&1.00&&3\%&&6.2&&100\%&0\%&&64.7&?&&6.2&?&&0.0&?&&---&33.5\\
35&20\%&1.00&&?&&?&&92\%&0\%&&?&?&&?&?&&?&?&&29.6&130071.1\\
35&30\%&1.00&&?&&?&&48\%&0\%&&?&?&&?&?&&?&?&&60.1&2990104.6\\
\bottomrule
\end{tabular}
\caption{Comparison of the PBB and MIP solvers, each being given one hour per instance}
%.  Each result is averaged over 50 experiments conducted for various instances.}
\label{tab:results}
\end{table*}

Main computational results are reported in Table~\ref{tab:results}.
The columns $k$, $d$ and $\alpha$ indicate the parameters of the instances.
For each combination of parameters, 100 instances were generated.
The column ``Sat.''\ reports the percentage of the instances that are satisfiable.
The column $w(\pi)$ shows the average weight of the optimal complete plans.
The other columns compare the MIP-based solver to the PBB algorithm.
Each of them is given one hour for each instance.
%%If the solver fails to terminate within that time (by which we mean finding and proving the optimality of the solution), the attempt counts as a failure.
The ``Solved'' columns show the percentage of instances successfully solved within the one hour limit by each of the solvers.
The ``Time, sec'' columns show the average time taken by each of the approaches.
If at least one of the runs failed for a solver, a question mark is shown in the corresponding cell of the table.
The $w_C(\pi)$ and $w_A(\pi)$ columns show the components of the weight corresponding to the constraints and the authorisations penalties, respectively.
For those parameters where at least one of the runs failed, we use the ``Best $w(\pi)$'' columns to to report the average weight of the best plan obtained by each of the solvers.

For each $k$, Table~\ref{tab:results} includes a range of instances starting from lightly constrained instances, which are mostly satisfiable, to highly constrained instances, none of which is satisfiable.
Naturally, the most interesting instances from the perspective of \vwsp are those that are unsatisfiable (since it is necessary to find an optimal plan of non-zero weight for such instances).
We are most interested in the unsatisfiable instances with moderate weights of the optimal complete plans.
A small weight $w(\pi)$ indicates that only a few minor exceptions are needed to implement the complete plan $\pi$.
With such a plan, it is easy to identify the bottleneck of the problem and refine it or accept the exceptions to the constraints as the exceptions are likely to be mild.
The $w_C(\pi)$ and $w_A(\pi)$ columns show that in most of the cases the authorisations were not broken.
In fact, there were only a few highly-constrained instances in which the the optimal complete plans assigned some steps to consultants, as the penalty for doing that is relatively high in our instances.

The complexity of the instances depends to a great extent on the number of steps $k$ and the parameters of the instances.
While small lightly-constrained instances can be easily tackled by either of the solvers, other instances clearly require an efficient algorithm.
The MIP solver succeeds with all the instances of size $k = 20$ but fails to solve many of the larger instances within an hour.
The PBB algorithm demonstrates a much better performance, solving all the instances of size up to $k = 30$ and the majority of the instances of size $k = 35$.
It is worth noting that the running time of the MIP solver can reach 10 minutes for $k = 20$ while the PBB solver solves all such instances within a fraction of a second.

Exact algorithms for solving hard optimisation problems will, necessarily, take a long time to compute results for certain instances.  
However, such an algorithm may find an optimum or near-optimum result long before the whole
solution space has been searched and can thus be used to compute a reasonable solution for instances that do not run to completion.  
The Best $w(\pi)$ column in Table~\ref{tab:results} clearly shows that MIP is far less suitable than PBB for this purpose.

To establish the practical limit on the problem size that each of the solvers can tackle within a reasonable time, we conducted another experiment to determine the number of instances that the two solvers could solve given at most one hour for each instance.  Figure~\ref{fig:solved} shows the results of the experiment.  Each result is averaged over 100 experiments for each instance.

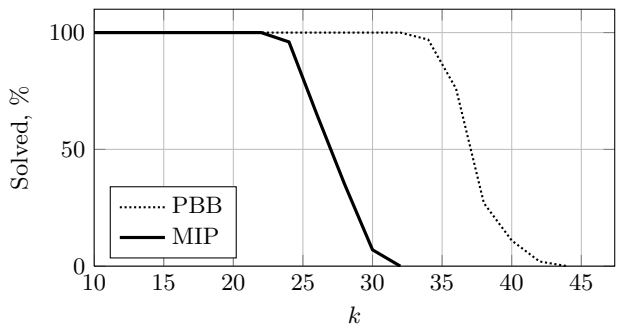
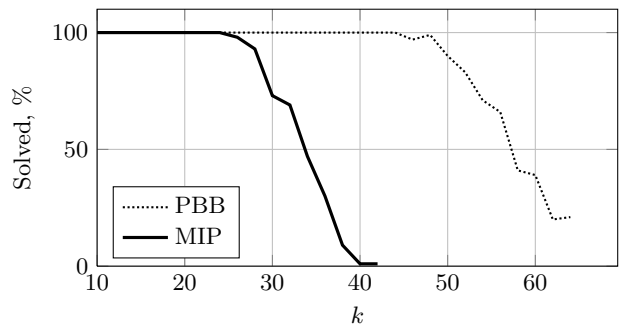
\begin{figure*}[htb]
    \centering
    \begin{subfigure}[t]{0.48\textwidth}
        \centering
        \begin{tikzpicture}
            \begin{axis}[
                compat=newest,
                width=\textwidth,
                height=5cm,
                xlabel={$k$},
                ylabel={Solved, \%},
                ymin=0,
                xmin=10,
                legend cell align=left,
                legend pos=south west,
                legend style={/tikz/every even column/.append style={column sep=0.5em}},
                grid=major
            ]
            \addplot[solid, densely dotted, thick, black] coordinates {
                (6, 100)
                (8, 100)
                (10, 100)
                (12, 100)
                (14, 100)
                (16, 100)
                (18, 100)
                (20, 100)
                (22, 100)
                (24, 100)
                (26, 100)
                (28, 100)
                (30, 100)
                (32, 100)
                (34, 97)
                (36, 76)
                (38, 27)
                (40, 11)
                (42, 2)
                (44, 0)
            };
            \addlegendentry{PBB}
            \addplot[solid, very thick, black] coordinates {
                (10, 100)
                (12, 100)
                (14, 100)
                (16, 100)
                (18, 100)
                (20, 100)
                (22, 100)
                (24, 96)
                (26, 65)
                (28, 35)
                (30, 7)
                (32, 0)
            };
            \addlegendentry{MIP}
            \end{axis}
        \end{tikzpicture}
        \caption{$d = 20\%$ and $\alpha = 1.00$}
        \label{fig:solved1}
    \end{subfigure}%
\hspace{0.04\textwidth}%
    \begin{subfigure}[t]{0.48\textwidth}
        \centering
        \begin{tikzpicture}
            \begin{axis}[
                compat=newest,
                width=\textwidth,
                height=5cm,
                ymin=0,
                xmin=10,
                xlabel={$k$},
                ylabel={Solved, \%},
                legend cell align=left,
                legend pos=south west,
                legend style={/tikz/every even column/.append style={column sep=0.5em}},
                        grid=major,,
            ]
            \addplot[solid, densely dotted, thick, black] coordinates {
                (6, 100)
                (8, 100)
                (10, 100)
                (12, 100)
                (14, 100)
                (16, 100)
                (18, 100)
                (20, 100)
                (22, 100)
                (24, 100)
                (26, 100)
                (28, 100)
                (30, 100)
                (32, 100)
                (34, 100)
                (36, 100)
                (38, 100)
                (40, 100)
                (42, 100)
                (44, 100)
                (46, 97)
                (48, 99)
                (50, 90)
                (52, 83)
                (54, 71)
                (56, 66)
                (58, 41)
                (60, 39)
                (62, 20)
                (64, 21)
            };
            \addlegendentry{PBB}
            \addplot[solid, very thick, black] coordinates {
                (10, 100)
                (12, 100)
                (14, 100)
                (16, 100)
                (18, 100)
                (20, 100)
                (22, 100)
                (24, 100)
                (26, 98)
                (28, 93)
                (30, 73)
                (32, 69)
                (34, 47)
                (36, 30)
                (38, 9)
                (40, 1)
                (42, 1)
            };
            \addlegendentry{MIP}
            \end{axis}
        \end{tikzpicture}

        \caption{$d = 10\%$ and $\alpha = 0.75$}
        \label{fig:solved2}
    \end{subfigure}

\caption{Comparison of the methods in terms of the ability to solve an instance within one hour}
%.  Each algorithm attempted to solve 50 instances of each size.}
\label{fig:solved}
\end{figure*}

Figure~\ref{fig:solved1} shows the performance of the methods on highly constrained instances.
Being given one hour, PBB solves 100\% of the instances of size up to $k = 32$.
In contrast, MIP can only reliably manage instances for $k \le 22$, and for $k = 32$ it fails to solve any instances at all.
Figure~\ref{fig:solved2} reports the results of the same experiment but for less constrained instances.
The results are broadly similar, with PBB solving all the instances of up to $k = 44$, whereas MIP fails for some instances when $k > 24$.
This experiment shows that the PBB algorithm significantly extends the range of solvable instances of \vwsp, something that will be important for large real-world workflow specifications.
Considering that the running time of each of the methods grows exponentially with the size of the problem (see appendix), large instances of \vwsp would require enormous computational power to be solved with MIP, while the PBB algorithm tackles them within minutes on a regular machine.

\section{Related Work}\label{sec:related-work}

Recent work on workflow satisfiability has borrowed techniques from the literature on constraint satisfaction~\cite{CoCrGaGuJo14}.
Indeed, WSP may be regarded as a constraint satisfaction problem, albeit with some unusual features which makes the study of WSP of interest in its own right.
Recent work in the constraint satisfaction community has made a distinction between ``hard'' and ``soft'' constraints: the former must be satisfied, while the latter may be broken provided the ``cost'' of breaking the constraint is taken into account.

The \emph{valued constraint satisfaction problem}, or VCSP for short, was introduced by Schiex, Fargier and Verfaillie~\cite{Schiex95} as a unifying framework for studying constraint programming with soft constraints. 
The study of a special case of VCSP, called {\em finite-valued} VCSP, was initiated by Cohen \etal~\cite{Cohen06}. 
In this case, useful for many applications, all weights are in $\mathbb{Z}$ (i.e., finite) and the objective function is the sum of appropriate weights. 
Valued CSP has influenced our framework for defining costs and \vwsp.

Recent work on WSP introduced the notion of a pattern for user-independent constraints, and bespoke algorithms, optimised to solve WSP using patterns, have been developed~\cite{CoCrGaGuJo14,KaGaGu}.
The branch-and-bound algorithm in Section~\ref{sec:branch-and-bound} is influenced by the work of Karapetyan, Gagarin and Gutin~\cite{KaGaGu}.
%Cohen \etal~\cite{Cohen06} asked for a classification allowing us to decide for which types of constraints the finite-valued VCSP is polynomial time solvable or for which it is NP-hard. Recently Thapper and \v{Z}ivny obtained such a dichotomy \cite{Thapper12,Thapper13}. 

The most closely related work in the literature on access control in workflows is that of Basin, Burri and Karjoth~\cite{BaBuKa12}, which considers the cost of modifying the authorisation policy when the workflow is unsatisfiable.
They encode the problem of minimizing this cost as a integer linear programming problem and use off-the-shelf software to solve the resulting problem.
We tackle the problem of an unsatisfiable workflow specification in a different way.
We assume the constraints and authorisation policy are fixed and instead associate costs with breaking the constraints and/or policies.
However, each violation will incur a cost and the goal of \vwsp is to minimise that cost.
Thus our approach provides greater flexibility than that of Basin \etal: we can break constraints as well as override the existing authorisation policy.
Obviously, there may be constraints (arising from statutory requirements, say) that cannot be broken.
Violation of such a constraint is simply assigned the maximum cost.
And of course, we can always refuse to implement a plan proposed by the algorithm.

Our work is also related to the growing body of research on risk-based and risk-aware access control~\cite{ChCr11,ChRoKeKaWaRe07,DiBeEyBaMo04,NiBeLo10}.
In such approaches, the decision returned by policy decision point for a given access request is not necessarily a simple ``allow'' or ``deny''.
The decision may be a number in the range $[0,1]$ indicating the risk associated with allowing the request, which allows the policy enforcement point to allow or deny the request on the basis of cumulative risk (either on a per-user or system basis).
The decision may also include an obligation that must be fulfilled by the policy enforcement point or requester to ensure that the risk is recorded and/or mitigated appropriately.

There is little work in the security literature on risk-aware workflows.
One exception is the MRARD framework of Han, Ni and Chen~\cite{HaNiCh09}.
However, the emphasis of their work (and of similar work in the business processing literature) is on the modelling and computation of risk, rather than determining an optimal assignment of users to steps in a workflow given the risk metrics.

\section{Concluding Remarks}\label{sec:conclusion}

We have established a framework that enables us to reason about unsatisfiable workflow specifications by associating costs with policy and constraint violations.
This, in turn, enables us to formulate the \vwsp, whose solution provides an assignment of users to steps that minimises the total cost of violations.
We have developed a bespoke algorithm for solving \vwsp and shown that its performance is far better than a generic solver, both in terms of the time taken to solve \vwsp and the range of instances that can be solved in a reasonable amount of time.

There are several interesting possibilities for future work.
One obvious possibility is to move to a completely risk-based approach for the assignment of users to steps in workflows.
Specifically, we retain the constraints but replace the authorisation policy with a risk matrix, associating each user-step pair with a risk.
The goal would be to ensure that the risk associated with a workflow instance remains below some specified threshold.

A second possibility arises from the idea of associating each pair $(T,u)$ with a cost, which provides the basis for an alternative ``non-linear'' approach to access control.
Suppose that we consider a set of permissions $P$, as in conventional role-based access control, and we associate a cost $\omega(Q,u)$ with each pair, where $u$ is a user and $Q$ is a subset of $P$.
Given an RBAC policy, expressed as a user-role relation $\ua \subseteq U \times R$ and permission-role relation $\pa \subseteq R \times P$, we write $P(u)$ to denote the set of permissions for which $u$ is authorised: that is, \mbox{$P(u) = \set{p \in P : \exists r \in R, (u,r) \in \ua, (r,p) \in \pa}$}. Then we define the weight of the policy to be
 \[
  w_A(\ua,\pa) = \sum_{u \in U} \omega(P(u),u).
 \]
This then raises some interesting questions that may have practical value.
We might, for example, consider the following problem: given inputs $U$, $P$, $\set{\omega(Q,u): u \in U, Q \subseteq P}$ and integer $k$, compute a set of roles of $R$ of cardinality $k$ and relations $\ua \subseteq U \times R$ and $\pa \subseteq R \times P$ such that at least one user is authorised for every permission and $w_A(\ua,\pa)$ is minimised.
Alternatively, we may insist that a user session does not exceed a ``budget'', where the cost of a session in which user $u$ invokes permissions $Q$ is defined to be $\omega(Q,u)$.

\bibliography{refs}
\bibliographystyle{acm}

% \clearpage

\appendix\section*{Running Times of the Algorithms}

Figure~\ref{fig:runtime} shows the running times of the two algorithms for two different choices of instance parameters.
(Each point on the graphs is averaged over 20 instances.)
Plotted on a log scale, the approximately linear growth in the the running times of the algorithms clearly demonstrates that the time required to solve \vwsp grows exponentially as the number of steps $k$ increases.
The figure also shows that the speed of growth is similar for both methods.  
However, the figure also clearly illustrates that the PBB algorithm outperforms the MIP solver by several orders of magnitude.

\begin{figure}[htb]
\begin{subfigure}{\columnwidth}
\begin{tikzpicture}
	\begin{semilogyaxis}[
		compat=newest,
		width=\columnwidth,
		height=7cm,
		xlabel={$k$},
		ylabel={Running time, sec},
		legend cell align=left,
		legend pos=south east,
		grid=major,
		%area legend,
	]
	\addplot[solid, densely dotted, thick, black] coordinates {
		(14, 0.01265155)
		(16, 0.0319072)
		(18, 0.04481565)
		(20, 0.11221665)
		(22, 0.2429501)
		(24, 0.6047307)
		(26, 2.5067187)
		(28, 8.1764823)
		(30, 32.8055693)
		(32, 96.6746861)
		(34, 363.30923255)
		(36, 6915.2128261)
	};
	\addlegendentry{PBB}
	\addplot[solid, very thick, black] coordinates {
		(10, 2.6738147)
		(12, 4.29069695)
		(14, 10.4197505)
		(16, 70.731748)
		(18, 188.93698255)
		(20, 535.4880369)
		(22, 756.1522563)
		(24, 1594.74560965)
	};
	\addlegendentry{MIP}
	\end{semilogyaxis}
\end{tikzpicture}
\caption{$d = 20\%$, $\alpha = 1.00$}
\label{fig:runtime1}
\end{subfigure}

\vspace*{1.5\baselineskip}

\begin{subfigure}{\columnwidth}%[htb]
\begin{tikzpicture}
	\begin{semilogyaxis}[
		compat=newest,
		width=\columnwidth,
		height=6cm,
		xlabel={$k$},
		ylabel={Running time, sec},
		legend cell align=left,
		legend pos=south east,
		grid=major,
		%area legend,
	]
	\addplot[solid, densely dotted, thick, black] coordinates {
		(14, 0.00984225)
		(16, 0.0048627)
		(18, 0.0175306)
		(20, 0.00935345)
		(22, 0.04034725)
		(24, 0.03335325)
		(26, 0.0193532)
		(28, 0.03612955)
		(30, 0.1789187)
		(32, 0.09102915)
		(34, 0.0735934)
		(36, 0.32155715)
		(38, 20.46105735)
		(40, 7.0337532)
		(42, 3.09938665)
		(44, 37.1631886)
		(46, 1490.68014615)
		(48, 19.84727945)
		(50, 10552.89387465)
	};
	\addlegendentry{PBB}
	\addplot[solid, very thick, black] coordinates {
		(10, 0.9932793)
		(12, 1.99175035)
		(14, 4.74861595)
		(16, 9.3575162)
		(18, 54.4678462)
		(20, 58.00983415)
		(22, 242.71836495)
		(24, 345.87222045)
		(26, 758.50626275)
		(28, 1933.9876331)
		(30, 5487.73406835)
	};
	\addlegendentry{MIP}
	\end{semilogyaxis}
\end{tikzpicture}
\caption{$d = 10\%$, $\alpha = 0.75$}
\label{fig:runtime2}
\end{subfigure}
\caption{Run-times of PBB and MIP as a function of $k$}\label{fig:runtime}
\end{figure}
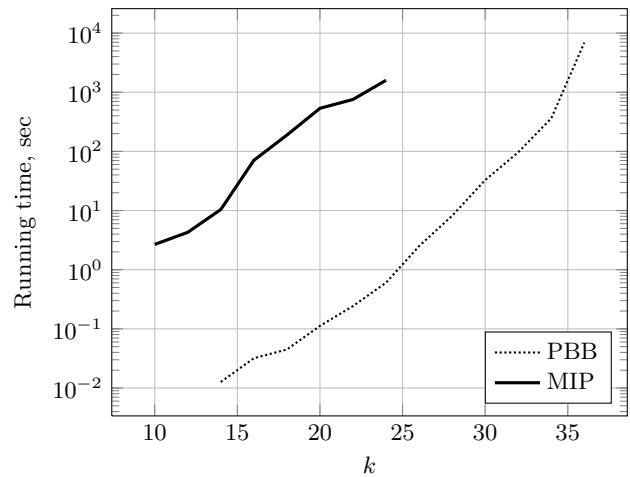
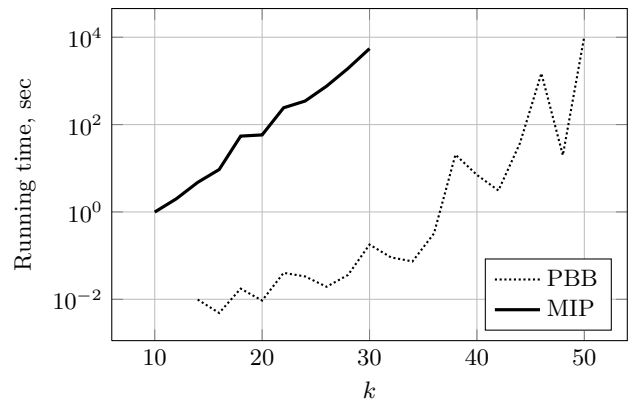

\end{document}